\documentclass[10pt]{article}
\usepackage{amsmath,amsthm,amsfonts,amssymb,url,authblk}
\newtheorem{theorem}{Theorem}[section]
\newtheorem{lemma}[theorem]{Lemma}

\newtheorem{example}{Example}[section]

\newcommand{\M}{\ensuremath{\mathcal{M}}} 
\newcommand{\B}{\ensuremath{\mathcal{B}}} 
\newcommand{\zed}{\ensuremath{\mathbb{Z}}} 
\newcommand{\SSS}{\ensuremath{\mathcal{S}}} 
\newcommand{\K}{\ensuremath{\mathcal{K}}} 
\newcommand{\Rec}{\ensuremath{\mathsf{Reconstruct}}} 
\newcommand{\Prob}{\ensuremath{\mathsf{Prob}}} 

\title{On the equivalence of authentication codes and robust $(2,2)$-threshold schemes\thanks{This research was supported by a ``Research in Pairs'' grant from the London Mathematical Society.}}
\author[1]{Maura B.\ Paterson}
\author[2]{Douglas R.\ Stinson%
\thanks{D.R.\ Stinson's research is supported by  NSERC discovery grant RGPIN-03882.}}
\affil[1]{Department of Economics,
Mathematics and Statistics\\ Birkbeck, University of London,
Malet Street, London WC1E 7HX, UK}
\affil[2]{David R.\ Cheriton School of Computer Science, University of Waterloo,
Waterloo, Ontario, N2L 3G1, Canada}

\begin{document}
\maketitle

\begin{abstract}
In this paper, we show a ``direct'' equivalence between certain authentication codes and robust secret sharing schemes. It was previously known that authentication codes and robust secret sharing schemes are closely related to similar types of designs, but direct equivalences had not been considered in the literature. Our new equivalences motivate the consideration of a certain ``key-substitution attack.'' We study this attack and analyze it in the setting of ``dual authentication codes.'' We also show how this viewpoint provides a nice way 
to prove properties  and generalizations of  some known constructions.
\end{abstract}

\section{Introduction}

We begin by giving some relevant definitions for authentication codes and threshold schemes. We should emphasize that we are only considering unconditionally secure cryptographic protocols in this paper. 

\subsection{Authentication Codes}

We follow Simmons' model for authentication \cite{Simmons}. A key $K$ determines an \emph{encoding rule} $e_K$, which is a possibly randomized mapping $e_K : \SSS \rightarrow \M$. Elements in $\SSS$ are called \emph{sources} and elements of $\M$ are \emph{messages}.
In general, we view $e_K(s)$ as a set of messages in the case that encoding is randomized (this is often called \emph{authentication with splitting}). The key $K$ is chosen from a \emph{keyspace} $\K$. The key and the source can be treated as independent random variables. If $|e_K(s)| = c$ for all $s \in \SSS$ and all $K \in \K$, the code is called a \emph{$c$-splitting} authentication code. 

Given a key $K$ and a message $m$, at most one source should be ``possible.'' 
That is, for every key $K$, we require that 
$e_K(s) \cap e_K(s') = \emptyset$ if $s\neq s'$. This ensures that the receiver, who has the key $K$ and a message $m \in e_K(s)$,
can uniquely determine the source $s$.

For any key $K$, denote  \[\mu(K) = \bigcup_s e_K(s).\] The set $\mu(K)$ consists of all the  messages that are  valid encodings of a source under key $K$. 

Also, for any message $m$, denote  \[\kappa(m) = \{ K : m \in \mu(K)\}.\] The set $\kappa(m)$ consists of all the keys for which $m$ is a valid encoding of some source.

The \emph{encoding matrix} of an authentication code is a matrix $E$ in which the rows are indexed by the keys in $\K$ and the columns are indexed by the sources in $\SSS$. The entry $E(k,s)$ is simply the set (of messages) $e_K(s)$. The entries $E(k,s)$ in the encoding matrix are singletons if and only if the code has no splitting, i.e., $c=1$.

We are primarily interested in authentication codes having \emph{perfect secrecy}, i.e.,  codes having the property that a message reveals no information about the source to an adversary who does not know the key. We also often want authentication codes that are secure against both message-substitution and key-substitution attacks. These attacks are defined as follows.

In a \emph{message-substitution attack} (also called a \emph{substitution attack}), the adversary sees a message $m$ and replaces it with a message $m' \neq m$. The adversary wins if
$m' \in e_K(s')$ and $m \in e_K(s)$, where $K$ is the (unknown) secret key and $s' \neq s$.
This is often just called a \emph{substitution attack}; these types of attacks have been considered for many years.

In a \emph{key-substitution attack}, the adversary sees a key $K$ and replaces it with a key $K' \neq K$. 
The adversary wins if
$m \in e_K(s)$ and $m \in e_{K'}(s')$, where $m$ is the (unknown) message and $s' \neq s$.
 This is perhaps a less natural type of attack to consider than a message-substitution attack. In fact, we are not aware of any previous study of unconditionally secure authentication codes that considers key-substitution attacks.
 However, we should note that a similar attack has been studied in the setting of systematic algebraic manipulation detection (AMD) codes; see \cite{Cramer}.
 
There is another attack that is often studied for authentication codes, namely, an \emph{impersonation attack}. 
In this attack, the adversary chooses a message, without seeing a ``previous'' message, hoping that it is an encoding of some source under the (unknown) key $K$.

The \emph{success probability} of an attack (impersonation, message-substitution or key-substitution) is the probability that the adversary wins the corresponding ``game.'' This probability is computed over a random choice of key, source, and message encoding  (if encoding of messages is randomized, i.e., in the case of a code with splitting) according to the  probability distributions 
specified on them. Throughout this paper, we assume that the  probability distributions defined on the keys and message encodings are uniform. That is, $\Prob [K] = 1/b$ for all keys $K \in \K$, and in a $c$-splitting code, for a given key $K \in \K$ and source $s \in \SSS$, we have $\Prob [m] = 1/c$ for all $m \in e_K(s)$.

Source probability distributions are often (but not always) assumed to be uniform.
Also, in all of the attacks we study, we assume that a key is used to encode only one message.

The adversary's optimal success probability for an impersonation attack is often denoted by $P_{d_0}$ and their 
optimal success probability for a message-substitution attack is denoted by $P_{d_1}$.
For $c$-splitting authentication codes, the following bounds are known.

\begin{theorem}
\label{bounds.thm}
\cite{Simmons,DeS,Blun}
Suppose a $c$-splitting authentication code for $k$ sources has $v$ messages. Then $P_{d_0} \geq ck/v$ and
$P_{d_1} \geq c(k-1)/(v-1)$.
\end{theorem}

For a code without splitting (i.e., $c=1$), the bounds obtained from Theorem \ref{bounds.thm} are 
$P_{d_0} \geq k/v$ and
$P_{d_1} \geq (k-1)/(v-1)$. This bound on $P_{d_1}$ was first proved by Massey \cite{Mas}.

%We now state a simple and useful way to show that certain authentication codes have perfect secrecy and minimum possible impersonation probabilities.

The following two results will be used several times later in the paper. They do not assume equiprobable sources. 

\begin{lemma}
\label{PD0.lem}
Suppose a $c$-splitting authentication code for $k$ sources has $v$ messages, $b$ equiprobable keys and equiprobable message encoding. Then 
$P_{d_0} = ck/v$ if and only if $| \kappa(m) | = bck/v$ for all messages $m$.
\end{lemma}

\begin{proof}
Let $K$ be the key that was chosen by the sender/receiver. The message $m$ chosen by the attacker 
will be accepted as valid if and only if 
$K \in \kappa(m)$. Since there are $b$ possible keys, the choice of $m$ will be a successful impersonation with probability 
$| \kappa(m) | / b$. 

Define \[A = \{ (K,m) : K \in \kappa(m)\}.\] Clearly 
\[|A| = \sum_{m \in \M} | \kappa(m) |.\]
However, we also have 
\[ |A| = \sum_{K \in \K } | \mu(K) | = bck.\] Therefore, 
\[ \max \{ | \kappa(m) | : m \in \M \} \geq \frac{bck}{v},\]
and equality occurs if and only if $| \kappa(m) | = bck/v$ for all messages $m$.

The adversary's optimal attack is to choose $m$ so that $| \kappa(m) |$ is maximized. Hence, the maximum success probability of an impersonation attack is at least $ck/v$, and equality occurs if and only if $| \kappa(m) | = bck/v$ for all messages $m$.
\end{proof}

\begin{theorem}
\label{secrecy.lem}
Suppose a $c$-splitting authentication code for $k$ sources has $v$ messages, $b$ equiprobable keys and equiprobable message encoding. Consider the following three conditions: 
\begin{enumerate}
\item $P_{d_0} = ck/v$;
\item the code achieves perfect secrecy; 
\item within each column of the encoding matrix, every message occurs the same number of times. 
\end{enumerate}
Then the code satisfies conditions 1. and 2. if and only if it satisfies condition 3. 
\end{theorem}

\begin{proof}
For any $s \in \SSS$ and  $m \in \M$, define
\[ \kappa(m,s) = \{ K: m \in e_K(s) \} .\]
Thus $\kappa(m,s)$ contains the keys for which $m$ is a valid encoding of $s$.
We  observe that
\begin{equation}
\label{ms.eq}
\Prob [ m \mid  s]  = \frac{| \{ K  : m \in e_K(s)\} |}{b} = \frac{| \kappa(m,s) |}{b}. 
\end{equation}

\medskip

 Suppose the code satisfies 2. Perfect secrecy is achieved if and only if  \[\Prob [ s \mid  m] = \Prob [ s] \] for all $s \in \SSS$ and all $m \in \M$. 
By Bayes' Theorem, this is equivalent to proving 
\begin{equation}
\label{PS.eq}
\Prob [ m \mid  s] = \Prob [ m] 
\end{equation} for all $s \in \SSS$ and all $m \in \M$.

For a given message $m$,  equations (\ref{ms.eq}) and (\ref{PS.eq})  imply that $| \kappa(m,s_1) | = | \kappa(m,s_2) |$ for all sources
$s_1,s_2$. 

It is clear that 
\[ \kappa(m) = \bigcup_{s \in \SSS}  \kappa(m,s),\]
where the sets $\kappa(m,s)$ ($s \in \SSS$) are disjoint. 
Therefore
\begin{equation}
\label{kappasum.eq}
 |\kappa(m)| = \sum _{s \in \SSS} |\kappa(m,s)| = k \times |\kappa(m,s)|
 \end{equation}
for any fixed source $s \in \SSS$. 

Now, assume additionally that the code satisfies condition 1. From Lemma \ref{PD0.lem}, we have $|\kappa(m)| = bck/v$.
Hence,  it follows that
\begin{equation}
\label{kappa.eq}
|\kappa(m,s)| = \frac{bc}{v}
\end{equation}
for every $m \in \M$, $s \in \SSS$. Therefore, condition 3.\ holds.

\medskip

Conversely, suppose condition 3.\ holds. Let $s \in \SSS$.
Then $|\kappa(m_1,s)| = |\kappa(m_2,s)|$ for all $m_1,m_2 \in \M$.
We have
 \[ \sum _{m \in \M} |\kappa(m,s)| = bc,\]
since there are $b$ rows in the encoding matrix and each cell contains $c$ messages.  
Therefore, equation (\ref{kappa.eq}) holds 
 for all $s \in \SSS$ and all $m \in \M$.
 Hence, from (\ref{ms.eq}), 
 \[\Prob [ m \mid  s]  = \frac{ \frac{bc}{v} }{b} = \frac{c}{v}\]
 for all $s \in \SSS$ and all $m \in \M$.
 Then
\begin{eqnarray*}
\Prob [ m ]  &=& \sum_{s \in \SSS }  (\Prob [ m \mid  s] \times  \Prob [ s]) \\
& = & \sum_{s \in \SSS }  \left( \frac{c}{v} \times  \Prob [ s] \right) \\
& = & \frac{c}{v} ,
\end{eqnarray*}
so  
$\Prob [ m \mid  s] = \Prob [ m] $ for all $s \in \SSS$ and  $m \in \M$.
Therefore, equation (\ref{PS.eq}) holds and we have perfect secrecy.

To see that $P_{d_0} = ck/v$, we use equation (\ref{kappasum.eq}), which is satisfied because we have perfect secrecy. 
Since (\ref{kappa.eq}) holds , we have
\[ | \kappa(m) | = k \times \frac{bc}{v} = \frac{bck}{v} \] for all $m \in \M$. Then
$P_{d_0} = ck/v$ from Lemma \ref{PD0.lem}. 
\end{proof}

\subsection{Threshold Schemes}

A  $(2,2)$-threshold scheme enables a \emph{secret} $s$ to be ``split'' into two \emph{shares} $v_1$ and $v_2$ in such a way that 
\begin{enumerate}
\item $v_1$ and $v_2$ uniquely determine $s$ via a \emph{reconstruction function}.
We express this as $\Rec(v_1,v_2)  = s.$
\item No individual share yields any information about the secret. That is, 
 \[\Prob [ s \mid  v_1] = \Prob [ s \mid v_2] =\Prob[s].\]
\end{enumerate}

More generally, a $(k,n)$-threshold scheme enables a secret $s$ to be split into $n$ shares in such a way that any $k$ shares permit the secret to be reconstructed, but no set of $k-1$ or fewer shares yield any information about the secret. 

In a \emph{robust} $(2,2)$-threshold scheme, we consider the scenario where one player may modify their share, hoping that $\Rec$ will then yield an incorrect secret. So we consider a setting where
$\Rec$  either returns a secret or $\perp$, where the latter indicates that no secret can be reconstructed 
from the two given shares. Suppose that the first player, $P_1$, alters their share as $v_1 \rightarrow v_1'$ ($P_1$ does not have any information about the value of the other share, $v_2$). Suppose
$\Rec(v_1,v_2) =s$. Then $P_1$ \emph{wins} this deception game if \[\Rec(v_1',v_2) = s'\] where $s' \neq s$.
$P_1$ \emph{loses} the game if \[\Rec(v_1',v_2) = s \quad \text{or} \quad \Rec(v_1',v_2) = \; \perp .\]
Similarly, if $P_2$ alters their share as $v_2 \rightarrow v_2'$, then they win the deception game if $\Rec(v_1,v_2') = s'$ where $s' \neq s$.

A \emph{robust} $(2,2)$-threshold scheme is \emph{$\epsilon$-secure} if no strategy by $P_1$ or $P_2$ will allow  them to win the deception game with probability exceeding $\epsilon$.  

\subsection{Background and Our Contributions}
\label{background.sec}

There has been previous work, for example in \cite{OKSS,OKS}, discussing constructions for ``optimal'' authentication codes and robust (threshold) secret sharing schemes using combinatorial structures such as BIBDs, difference sets, external BIBDs (EBIBDs), external difference families (EDFs) and splitting BIBDs. 
In the context of authentication codes,  ``optimal'' means that the deception probabilities are as small as possible and the number of encoding rules (or keys) 
is also as small as possible. For a robust threshold scheme, ``optimal'' means that the deception probabilities meet a specified bound that is expressed in terms of the number of possible shares and the number of possible secrets. Additionally, \cite{OKSS,OKS} proved some partial converses showing that optimal authentication codes and robust secret sharing schemes imply the existence of some of the above-mentioned combinatorial structures. Without going into details, the following are the main results along this line:
\begin{itemize}
\item In \cite{OKSS} it is shown that a robust threshold scheme can be constructed from an EDF with 
$\lambda = 1$. (This construction incorporates a Shamir threshold scheme as an ingredient.) Conversely, certain robust threshold schemes give rise to certain EBIBDs.
\item In \cite{OKS} it is shown that a robust threshold scheme can be constructed from a difference set. (This construction also incorporates a Shamir threshold scheme as an ingredient.) Conversely, certain robust threshold schemes give rise to certain SBIBDs.
\item In \cite{OKSS}, a construction is given for splitting authentication codes from EDFs with $\lambda = 1$. This paper also constructs certain authentication codes from splitting BIBDs with $\lambda =1$, as well as proving a converse result.
\end{itemize}

The above results suggest that there are connections between authentication codes and robust secret sharing schemes, as they are closely related to similar (and sometimes identical) types of designs. For example, 
the combinatorial designs generated by EDFs (which include difference sets as a special case) can be used to construct both robust threshold schemes and authentication codes.

In this paper, we show a ``direct'' equivalence between certain authentication codes and robust secret sharing schemes. We also study a key-substitution attack for authentication codes and interpret it in light of what we term ``dual authentication codes.''

Robust $(k,n)$-threshold schemes were introduced by Tompa and Woll \cite{TW}.
They have been constructed in the past by a two-step process: 
First, the secret is ``encoded'' using a suitable combinatorial structure such as a difference set \cite{OKS}, 
EDF \cite{OKSS} or AMD code \cite{Cramer}. Second, the encoded secret is shared using a traditional Shamir threshold scheme. However, if we consider a $(2,2)$-threshold scheme, then the second step is not required and consequently we can show a direct equivalence between authentication codes and $(2,2)$-threshold schemes.

The rest of this paper is organized as follows.
In Section \ref{equiv.sec}, we prove our main equivalence result. 
Constructions for authentication codes that satisfy the necessary hypotheses are studied in Section \ref{constructions.sec}.
The notion of ``dual authentication codes'' is introduced and explored in Section \ref{dual.sec}. Finally, some closing remarks are given in Section \ref{summary.sec}.

\section{Equivalences}
\label{equiv.sec}

In the next subsections, we show the equivalence of certain authentication codes and robust $(2,2)$-threshold schemes.

\subsection{Threshold Scheme to Authentication Code}

Given an $\epsilon$-secure robust $(2,2)$-threshold scheme, we construct an authentication code.
This is somewhat similar to the the construction used by Kurosawa, Obana and Ogata in \cite[Theorem 15]{KOO}.

For any ordered pair of shares $(v_1,v_2)$ such that $\Rec(v_1,v_2)  = s$, define $v_2 \in e_{v_1}(s)$.
First, we note that $e_{v_1}(s) \neq \emptyset$ for all $v_1$ and all $s$. This holds because the share $v_1$ 
does not provide any information about the secret. Hence, for all choices of $v_1$ and $s$, there must be at least one value $v_2$ such that $\Rec(v_1,v_2)  = s$.

The probability distribution on the sources in the authentication code 
should be the same as the probability distribution on the shares of the threshold scheme. 
Also, note the following correspondences: 
\begin{center}
\begin{tabular}{ccc}
threshold scheme & & authentication code\\ \hline
source $s$ & $\longleftrightarrow$ & secret $s$\\
share $v_1$ & $\longleftrightarrow$ & key $K$\\
share $v_2$ & $\longleftrightarrow$ & message $m$.
\end{tabular}
\end{center}
We show that the resulting authentication scheme satisfies various properties now.

\medskip

\noindent {\bf Message-substitution attack.}  Suppose an adversary replaces $m \in e_{K}(s)$ 
with $m' \neq m$ in the authentication code. 
This corresponds to modifying share $v_2$ (the second share)  from $m$ to $m'$ in the robust secret sharing scheme. 
Because the threshold scheme is robust, we know that 
\[\Prob[\Rec(K,m')  = s' \neq s] \leq \epsilon.\]
In other words, \[\Prob [m' \in e_{K}(s') \text{ and } s' \neq s] \leq \epsilon.\]
Therefore, the probability of a successful message-substitution attack is at most $\epsilon$.

\medskip

\noindent {\bf Key-substitution attack.}  Suppose an adversary replaces $K$ 
with $K' \neq K$ in the authentication code, where $m \in e_K(s)$.
This corresponds to modifying share $v_1$ (the first share) from $K$ to $K'$ in the robust secret sharing scheme. 
Because the threshold scheme is robust, we know that \[\Prob[\Rec(K',m)  = s' \neq s] \leq \epsilon.\]
In other words, \[\Prob [m \in e_{K'}(s') \text{ and } s' \neq s] \leq \epsilon.\]
Therefore, the probability of a successful key-substitution attack is at most $\epsilon$.

\medskip

\noindent {\bf Perfect Secrecy.} The threshold scheme has the property that one share yields no information about the value of the secret. Therefore, in particular, 
\[ \Prob [ s \mid v_2] =\Prob[s].\]
Suppose the share $v_2$ is fixed but we have no information about the share $v_1$.
Then we have  no information about the secret $s$. In the corresponding authentication code, this means 
that the message $m=v_2$ provides no information about the source $s$ when the key $K=v_1$ is not known, so we have perfect secrecy.

\medskip

It is also possible to construct authentication codes with similar properties from any robust 
$(k,n)$-threshold scheme with $k \geq 2$. For example, see \cite{OKS}. 
The idea is to fix shares for the first $k-2$ players, say, by choosing some $(k-2)$-tuple of shares that occurs with probability greater than $0$. Consider the subset of distribution rules such that the first $k-2$ shares take on the specified values. Retain the shares for the next two players, but throw away the  shares that would be given to the last $n-k$ players. This gives rise to a $(2,2)$-threshold scheme, which can then be used to construct an authentication code using the above-described technique.

\subsection{Authentication Code to Threshold Scheme}

The construction in the previous subsection can easily be reversed.  Now we start with an authentication code having perfect secrecy and we assume that message-substitution and key-substitution attacks have success probability at most $\epsilon$. We construct a $(2,2)$-threshold scheme as follows:  shares for $P_1$ are keys in the authentication code, shares for $P_2$ are messages in the authentication code, and secrets are sources in the authentication code.  Note that $P_1$ and $P_2$ have shares of the same size if and only if the number of keys is the same as the number of messages (in the authentication code).

For every $m \in e_K(s)$, construct a distribution rule $(K,m;s)$, i.e., $v_1 = K$, $v_2 = m$ and \[\Rec(v_1,v_2)  = \Rec(K,m)  = s.\] We need to show that the resulting set of distribution rules defines an 
$\epsilon$-secure $(2,2)$-threshold scheme.

\medskip

\noindent {\bf Secret reconstruction.}  Suppose that we have two distribution rules $(K,m;s)$ and $(K,m;s')$ with $s' \neq s$.  Then $m \in e_K(s) \cap e_K(s')$ in the authentication code, which is not allowed. Thus, two shares determine at most one secret.

\medskip

\noindent {\bf Information revealed by one share.}   We want to prove that 
\[ \Prob [ s \mid v_1] = \Prob [ s \mid v_2] = \Prob[s].\]
If $v_1 = K$ is given, then this yields no information about $s$ because $K$ and $s$ are independent in the authentication code. If $v_2 = m$ is given, then this yields no information about $s$ because the authentication code has perfect secrecy.

\medskip

\noindent {\bf Modifying $v_1$.}  Suppose  $P_1$ replaces their share $v_1= K$ 
with $v_1' = K' \neq K$. 
This corresponds to a key-substitution attack in the authentication code. We know that
\[\Prob [m \in e_{K'}(s') \text{ and } s' \neq s] \leq \epsilon,\] so 
\[\Prob[\Rec(K',m)  = s' \neq s] \leq \epsilon.\]

\medskip

\noindent {\bf Modifying $v_2$.}  Suppose  $P_2$ replaces their share $v_2= m$ 
with $v_2' = m' \neq m$. 
This corresponds to a message-substitution attack in the authentication code. We know that
\[\Prob [m' \in e_{K}(s') \text{ and } s' \neq s] \leq \epsilon,\] so 
\[\Prob[\Rec(K,m')  = s' \neq s] \leq \epsilon.\]

\subsection{Main Theorem}

Summarizing the results in the two previous subsections, we have our main equivalence theorem. For simplicity, we assume equiprobable distributions of sources (in the authentication code) and secrets (in the threshold scheme).

\begin{theorem} 
\label{main.thm}
There exists an authentication code with perfect secrecy for $k$ uniformly distributed sources that is $\epsilon$-secure against message-substitution and key-substitution attacks if and only if there exists
an $\epsilon$-secure $(2,2)$-threshold scheme for $k$ uniformly distributed secrets.
\end{theorem}

\section{Combinatorial Constructions}
\label{constructions.sec}

In this section, we look at various constructions for authentication codes that are based on combinatorial designs, paying particular attention to the properties (namely, perfect secrecy and key-substitution attacks) that are relevant for the construction of robust $(2,2)$-threshold schemes using Theorem \ref{main.thm}.  Throughout this section, we assume standard design-theoretic definitions that can be found, for example,  in \cite{HCD}.

\subsection{Symmetric BIBDs}
\label{SBIBDs.sec}

%We first consider constructions based on \emph{balanced incomplete block designs} (or BIBDs) in this section.

First, we give a simple construction using symmetric BIBDs (i.e., SBIBDs).
This is a slight generalization of constructions given in \cite{OKSS,OKS} since we do not require that the SBIBD is generated from a difference set.
 
Suppose that $(X, \B)$ is a $(v,k, \lambda)$-SBIBD (so $\lambda(v-1) = k(k-1)$). 
Suppose that $X = \{x_i : 1 \leq i \leq v\}$ is the set of points in the design and
$\B = \{B_j: 1 \leq j \leq v\}$ is the set of blocks in the design. 
We can order each block $B_j$ to obtain a $k$-tuple
$C_j = (c_{1,j}, \dots , c_{k,j})$ in such a way that the following property is satisfied:
\[ | \{j :  c_{\ell,j} \} = x_i| = 1 \] for every $i$, $1 \leq i \leq v$, and every $\ell$, $1 \leq \ell \leq k$. 
That is, we can write out the ordered blocks $C_j$ ($1 \leq j \leq v$) as the rows of a $v$ by $k$ array $E$ in such a way that every point occurs once in each column of the array $E$.  Such an array is known as a \emph{Youden square}; see, for example, \cite[\S VI.65]{HCD}.

A Youden square can be constructed from any SBIBD by using systems of distinct representatives. However, in the case where the SBIBD is generated from a difference set in an abelian group $G$, the Youden square occurs automatically if we arbitrarily order the base block and then generate the rest of the (ordered) blocks by developing the base block through the group $G$.

Suppose we use $E$ as an encoding matrix for an authentication code.
Thus, a key corresponds to a block in the design, or equivalently a row in $E$. The $k$ sources are the $k$ columns in $E$ and the messages are the $v$ points in the design. We assume that the sources are equiprobable.

It is not difficult to verify that this authentication code 
is $(k-1)/(v-1)$-secure against message-substitution and key-substitution attacks
(see the proof of Theorem \ref{BIBD.thm} for additional detail). It is also clear that this authentication code provides perfect secrecy; this follows immediately from Theorem \ref{secrecy.lem} using the ``Youden square'' property of the authentication matrix. 
This construction is in fact a special case of \cite[Theorem 5.5]{OKSS}, extended to include the perfect secrecy property by using an appropriate ordering of the blocks, as described above. 

Starting with this authentication code, we obtain from Theorem \ref{main.thm} an $\epsilon$-secure $(2,2)$-threshold scheme for $k$ equiprobable secrets, where $\epsilon = (k-1)/(v-1)$.
Summarizing, we have the following theorem.

\begin{theorem}
\label{SBIBD}
If there exists a $(v,k,\lambda)$-SBIBD, then there exists 
\begin{enumerate}
\item an authentication code with perfect secrecy for $k$ equiprobable sources that is $(k-1)/(v-1)$-secure against message-substitution and key-substitution attacks, and 
\item a $(k-1)/(v-1)$-secure $(2,2)$-threshold scheme for $k$ equiprobable secrets, in which the share sets for both players have size $v$.
\end{enumerate}
\end{theorem}

\begin{example}
\label{731.exam}
{\rm A $(7,3,1)$-SBIBD is just a projective plane of order 2, often called the Fano plane. The seven blocks in the design can be obtained from the base block $\{0,1,3\}$ by developing it in the group $\zed_7$. After ordering the blocks appropriately, we obtain the following Youden square.

\[
\begin{array}{|c|c|c|}
\hline
s_1 & s_2 & s_3 \\ \hline 
0 & 1 & 3\\ \hline
1 & 2 & 4\\ \hline
2 & 3 & 5\\ \hline
3 & 4 & 6\\ \hline
4 & 5 & 0\\ \hline
5 & 6 & 1\\ \hline
6 & 0 & 2\\ \hline
\end{array}
\]

This Youden square is the encoding matrix for an authentication code  with perfect secrecy having $P_{d_0} = 3/7$ and $P_{d_1} = 2/6 = 1/3$. The success probability of any key-substitution attack is also $1/3$. For example, if $K_1$ is replaced by $K_2$, then the attack succeeds if and only if $m=1$. The probability that $m=1$ (given that $K_1$ is the key) is $1/3$ because the sources are equiprobable.

The corresponding $(2,2)$-threshold scheme is $(1/3)$-secure and has the following $21$ distribution rules:
\[
\begin{array}{c@{\hspace{.5in}}c@{\hspace{.5in}}c}
\begin{array}{|c|c|c|}
\hline
v_1 & v_2 & s \\ \hline 
0 & 0 & s_1\\ \hline
1 & 1 & s_1\\ \hline
2 & 2 & s_1\\ \hline
3 & 3 & s_1\\ \hline
4 & 4 & s_1\\ \hline
5 & 5 & s_1\\ \hline
6 & 6 & s_1\\ \hline
\end{array}
&
\begin{array}{|c|c|c|}
\hline
v_1 & v_2 & s \\ \hline 
0 & 1 & s_2\\ \hline
1 & 2 & s_2\\ \hline
2 & 3 & s_2\\ \hline
3 & 4 & s_2\\ \hline
4 & 5 & s_2\\ \hline
5 & 6 & s_2\\ \hline
6 & 0 & s_2\\ \hline
\end{array}
&
\begin{array}{|c|c|c|}
\hline
v_1 & v_2 & s \\ \hline 
0 & 3 & s_3\\ \hline
1 & 4 & s_3\\ \hline
2 & 5 & s_3\\ \hline
3 & 6 & s_3\\ \hline
4 & 0 & s_3\\ \hline
5 & 1 & s_3\\ \hline
6 & 2 & s_3\\ \hline
\end{array}
\end{array}
\]
Any deception carried out by $P_1$ or $P_2$ succeeds with probability $1/3$. For example, suppose $v_1 \rightarrow v_1' = v_1 + 1 \bmod 7$. This deception will succeed if and only if $s = s_2$. In this case, $v_2 = v_1 + 1 \bmod 7$ and then $\Rec(v_1',v_2) = s_1$. The success probability of this deception is $\Prob[s  = s_2] = 1/3$ because the sources are equiprobable.
}
\end{example}

We should note that the authentication codes and robust threshold schemes obtained from Theorem \ref{SBIBD} are optimal in various senses. In the case of the authentication code, the impersonation and message-substitution attacks have success probability that is as small as possible, according to Massey's bounds \cite{Mas}.
Also, the number of encoding rules (or keys) is as small as possible, from \cite[Theorem 2.1]{RS}.
 
For the threshold schemes, we have $v$ possible shares, $k$ possible secrets, and 
the scheme is $\epsilon$-secure where $\epsilon = (k-1)/(v-1)$. This meets the bound proven in 
\cite[Corollary 3.3]{OKS}. In fact, as a result of our discussion above, we have shown the following strong characterization of these ``optimal''  $(2,2)$-threshold schemes. %Theorem \ref{iff.thm}  extends the results in \cite{OKS}.

\begin{theorem}
\label{iff.thm}
There exists a $(v,k,\lambda)$-SBIBD if and only if there exists a $(k-1)/(v-1)$-secure $(2,2)$-threshold 
scheme for $k$ equiprobable secrets.
\end{theorem}

\subsection{BIBDs}

More generally, we can use any BIBD (i.e, not necessarily a symmetric BIBD) to construct an authentication code.
It has also been shown that the resulting authentication codes can provide 
perfect secrecy if obvious numerical conditions are satisfied; for example, see \cite[Theorem 6.4]{St}.
Here is a ``classical'' construction of authentication codes from BIBDs.

\begin{theorem} 
\label{BIBD.thm}
Suppose there is a $(v,b,r,k,\lambda)$-BIBD where $r \equiv 0 \bmod k$. Then there is an authentication code for $k$ equiprobable sources, having $v$ messages and 
$b$ equiprobable keys, which satisfies the following properties:
\begin{enumerate}
\item $P_{d_0} = k/v$ and $P_{d_1} = (k-1)/(v-1)$, 
\item the code provides perfect secrecy, and
\item if $r=k$, then the optimal key-substitution attack has success probability $(k-1)/(v-1)$, and if $\lambda = 1$, then the optimal key-substitution attack has success probability $1/k$.
\end{enumerate}
\end{theorem}

\begin{proof}
First, we order each block in such a way that each element occurs exactly $r/k$ times in each position.
To do this, the technique used in the proof of \cite[Theorem 6.4]{St} can be applied (the proof of \cite[Theorem 6.4]{St} assumed $\lambda = 1$, but the method can be generalized easily to arbitrary $\lambda$). 
Then, Theorem \ref{secrecy.lem} shows that the resulting authentication code has perfect secrecy and $P_{d_0} = k/v$.

We now prove 3, which treats the special cases of (1) SBIBDs and (2) BIBDs with $\lambda = 1$.  First, we look at authentication codes derived from an SBIBD. The encoding matrix has one occurrence of each message in each column. Suppose we replace any key $K_i$ with any other key $K_j$. There are exactly $\lambda$ messages that occur in both $K_i$ and $K_j$, and each such message  occurs in a different position in $K_i$ and $K_j$. Thus the attack is successful if and only if the message $m$ is one of these $\lambda$ messages. The sources are equiprobable, so the success probability is 
$\lambda /k = (k-1)/(v-1)$.

Suppose now that the code is derived from a BIBD with $\lambda = 1$. Suppose the attacker replaces a key $K_i$ with another key $K_j$. Observe that there is at most one message that occurs in both $K_i$ and $K_j$. If $K_i$ and $K_j$
contain no common message, or if they contain a common message in the same column, the attack will not succeed. Therefore the attacker should choose $K_j$ so that $K_i$ and $K_j$ contain a common message that occurs in different columns. Given a message $m$ in row $K_i$, there are $r - r/k = r(k-1)/k$ rows in which $m$ occurs in a different column than it does in $K_i$. Since $\lambda = 1$ and there are $k$ messages in row $K_i$, the number of rows $K_j$ such that $K_i$ and $K_j$ contain a common message that occurs in different columns is precisely $kr(k-1)/k = r(k-1)$. The optimal attack is to choose one of these $r(k-1)$ rows; the success probability is  
\[ \frac{r(k-1)/k}{r(k-1)} = \frac{1}{k}.\]
\end{proof}

Computing the success probability of a key-substitution attack is, in general, more complicated, as blocks of a BIBD
might intersect in different numbers of points. There were two types of BIBDs considered in part 3 of Theorem \ref{BIBD.thm}. Suppose we then construct a robust $(2,2)$-threshold scheme from the authentication code using the transformation given in Section \ref{equiv.sec}. The success of modifying share $v_1$ is quantified by the success of the key-substitution attack in the authentication code setting, whereas the success of modifying share $v_2$ is the same as the success of the message-substitution attack  in the authentication code setting.
In general, the  success probabilities of the two share-modification attacks will be different; however, if we start with an SBIBD, the probabilities are the same. Theorem \ref{SBIBD} is fact just the specialization of Theorem \ref{BIBD.thm} to symmetric BIBDs.

Applying Theorem \ref{main.thm}, we have the following.

\begin{theorem}
\label{BIBD1.thm}
If there exists a $(v,k,1)$-BIBD, then there exists 
\begin{enumerate}
\item an authentication code with perfect secrecy for $k$ equiprobable sources that is $(1/k)$-secure against message-substitution and key-substitution attacks, and 
\item a $(1/k)$-secure $(2,2)$-threshold scheme for $k$ equiprobable secrets.
\end{enumerate}
\end{theorem}

\begin{proof}
We showed in Theorem \ref{BIBD.thm} that the authentication code arising from a $(v,k,1)$-BIBD is $(k-1)/(v-1)$-secure against message-substitution attacks and 
$(1/k)$-secure against key-substitution attacks. Since we have
\[ \frac{k-1}{v-1} \leq \frac{1}{k}\]
if a $(v,k,1)$-BIBD exists, 
the authentication code is $(1/k)$-secure against both attacks. Then the stated result follows directly from  Theorem \ref{main.thm}.
\end{proof}

We note that the $(2,2)$-threshold scheme arising from part 2.\ of Theorem \ref{BIBD1.thm} has share sets (for the two players) of different sizes, 
unless the BIBD is a projective plane.

\subsection{External difference families}

A construction for splitting authentication codes using external difference families (or EDFs) was given in \cite{OKSS}. 
%Here we interpret this construction in light of our approach using dual authentication codes.
First, we define EDFs. Let $G$ be an additive abelian group of order $n$ having identity $0$. An \emph{$(n,k,c,\lambda)$-external difference family}  is a set of $k$ $c$-subsets of $G$, say $D_1, \dots , D_k$, such that
the following multiset equation holds.
\[ \{ x - y : x \in D_i, y \in D_j, i \neq j \} = \lambda (G \setminus \{0\}) .\]
That is, when we look at the differences of elements from different $c$-subsets in the EDF, we see every non-zero value occurring exactly $\lambda$ times. Therefore, a necessary condition for existence of an $(n,k,c,\lambda)$-EDF is that the following equation holds:
\begin{equation}
\label{EDF.eq}
\lambda (n-1) = c^2 k (k-1).
\end{equation}

The following theorem is a straightforward generalization of \cite[Theorem 3.4]{OKSS}, 
which only addressed the case $\lambda = 1$ and did not explicitly discuss key-substitution attacks.

\begin{theorem} 
\label{EDF.thm}
Suppose there is an $(n,k,c,\lambda)$-EDF. 
Then there is a $c$-splitting authentication code $E$ for $k$ equiprobable sources, having $n$ messages and 
$n$ equiprobable keys, such that
\begin{enumerate}
\item the code provides perfect secrecy,
\item $P_{d_0} = ck/n$ and $P_{d_1} = c(k-1)/(n-1)$, and
\item the optimal key-substitution attack has success probability $c(k-1)/(n-1)$.
\end{enumerate}
%Further, the dual authentication code, $F$, satisfies the same properties. 
\end{theorem}

\begin{proof}
We first specify an arbitrary ordering of the $k$ $c$-subsets in the EDF and then we develop the EDF through the abelian group $G$, maintaining the same ordering (as is done in Example \ref{EDF.exam}). This yields the encoding matrix of a $c$-splitting authentication code. In each column of the encoding matrix, we see exactly $c$ occurrences of each element of $G$. From Theorem  \ref{secrecy.lem}, we have perfect secrecy and $P_{d_0} = ck/n$.

In a message-substitution attack, a message $m$ is substituted with $m'$. There are precisely $\lambda$ rows of $E$ that contain $m$ and $m'$ in different $c$-subsets; these are the keys for which the particular substitution will succeed. Also, there are $kc$ rows that contain $m$. Since the sources are equiprobable, the probability of a successful message substitution is
\[ \frac{\lambda}{kc} =  \frac{c(k-1)}{n-1},\] by applying (\ref{EDF.eq}).

For a key-substitution attack, a key $K$ is given to the attacker and the attacker must choose a different key $K'$. 
Because the encoding matrix is generated from an EDF, there is a value $d \in G$, $d \neq 0$, such that 
$e_{K'}(s) = e_{K}(s) + d$ for all $s$. 
From this fact, it is not hard to see that
\[ \{m:  m \in e_K(s) \cap e_{K'}(s'), s \neq s'\} 
=
\{ m : m \in e_K(s), m-d \in e_{K}(s'), s \neq s'\}.
\]
Hence, there are exactly $\lambda$ messages $m$ such that the attack where $K$ is replaced by $K'$ is successful. 
Since there are $kc$ possible messages $m \in \mu(K)$, and these values of $m$ are equally likely, the key-substitution attack has success probability 
$\lambda/(kc) = c(k-1)/(n-1)$.
\end{proof}

Observe that the values of $P_{d_0}$ and $P_{d_1}$ in Theorem \ref{EDF.thm} are optimal, by Theorem \ref{bounds.thm}. Also, we have shown that the optimal message-substitution and key-substitution attacks 
in the above-constructed code have the same success probability, namely
$c(k-1)/(n-1)$. Thus, if we apply Theorem \ref{main.thm}, 
we obtain a $c(k-1)/(n-1)$-secure $(2,2)$-threshold scheme for $k$ secrets.

\begin{example} 
\label{EDF.exam}
{\rm The three sets $\{1,7,11\}$, $\{4,7,9\}$, $\{5,16,17\}$ form a $(19,3,3,3)$-EDF in $\zed_{19}$.
We develop these sets modulo 19 to obtain the following encoding matrix for a $3$-splitting authentication code:
\[E =
\begin{array}{|c|c|c|}
\hline
s_1 & s_2 & s_3 \\ \hline 
\{1,7,11\} & \{4,6,9\} & \{5,16,17\} \\ \hline
\{2,8,12\} & \{5,7,10\} & \{6,17,18\} \\ \hline
\{3,9,13\} & \{6,8,11\} & \{7,18,0\} \\ \hline
\{4,10,14\} & \{7,9,12\} & \{8,0,1\} \\ \hline
\{5,11,15\} & \{8,10,13\} & \{9,1,2\} \\ \hline
\{6,12,16\} & \{9,11,14\} & \{10,2,3\} \\ \hline
\{7,13,17\} & \{10,12,15\} & \{11,3,4\} \\ \hline
\{8,14,18\} & \{11,13,16\} & \{12,4,5\} \\ \hline
\{9,15,0\} & \{12,14,17\} & \{13,5,6\} \\ \hline
\{10,16,1\} & \{13,15,18\} & \{14,6,7\} \\ \hline
\{11,17,2\} & \{14,16,0\} & \{15,7,8\} \\ \hline
\{12,18,3\} & \{15,17,1\} & \{16,8,9\} \\ \hline
\{13,0,4\} & \{16,18,12\} & \{17,9,10\} \\ \hline
\{14,1,5\} & \{17,0,3\} & \{18,10,11\} \\ \hline
\{15,2,6\} & \{18,1,4\} & \{0,11,12\} \\ \hline
\{16,3,7\} & \{0,2,5\} & \{1,12,13\} \\ \hline
\{17,4,8\} & \{1,2,6\} & \{2,13,14\} \\ \hline
\{18,5,9\} & \{2,4,7\} & \{3,14,15\} \\ \hline
\{0,6,10\} & \{3,5,8\} & \{4,15,16\} \\ \hline
\end{array}
\]
The rows of $E$ are indexed by $K_0, \dots , K_{18}$.
The optimal success probability of a message-substitution attack or  
a key-substitution attack is  $1/6$.
The code also has perfect secrecy.
}
\end{example}

An EDF also gives rise to a robust $(2,2)$-threshold scheme by applying Theorem \ref{main.thm}.
The two share sets in the threshold scheme have the same size because the authentication code derived from the EDF has the same number of messages as keys.

\begin{theorem}
If there exists an $(n,k,c,\lambda)$-EDF, then there exists 
 a $c(k-1)/(n-1)$-secure $(2,2)$-threshold scheme for $k$ equiprobable secrets, 
 in which the share sets for both players have size $n$.
\end{theorem}
 
\subsection{Splitting BIBDs}

Splitting BIBDs were defined in \cite{OKSS}. A \emph{$(v, u \times c,1)$-splitting BIBD} 
is a set system consisting of a set $X$ of  $v$ points and a set $\B$ of blocks of size $uc$, 
which satisfies the following properties:
\begin{enumerate}
\item each block $B$ can be partitioned into $u$ subsets of size $c$, which are denoted $B_i$, $1 \leq i \leq u$, and 
\item given any two distinct points $x$ and $y$, there is a unique block $B$ such that $x \in B_i$ and $y \in B_j$, where $i \neq j$.
\end{enumerate}

We note that $(v, u \times 1,1)$-splitting BIBD is the same thing as a $(v, u,1)$-BIBD.

A $(v, u \times c,1)$-splitting BIBD has replication number $r$ and $b$ blocks, where
\[ r = \frac{v-1}{(u-1)c} \quad \text{and} \quad 
b = \frac{vr}{uc} = \frac{v(v-1)}{u(u-1)c^2}.\]
Of course $r$ and $b$ must be integers if a  $(v, u \times c,1)$-splitting BIBD exists.

The following definition is new. 
%It is relevant for discussing perfect secrecy in the authentication codes derived from splitting BIBDs.
A $(v, u \times c,1)$-splitting BIBD is  \emph{equitably ordered} if the multiset equation 
\[ \bigcup_{B \in \B} B_i = \frac{r}{u} X\] is satisfied for all $i$, $1 \leq i \leq u$.
If a splitting BIBD is equitably ordered, then it yields an 
authentication code with perfect secrecy, from Theorem \ref{secrecy.lem}.

It is shown in \cite{PS} that a $(v, u \times c,1)$-splitting BIBD can be equitable ordered only if 
\begin{equation}
\label{v.eq} v \equiv 1 \bmod (u (u-1)c^2).\end{equation}

In the case $c=1$, where a splitting BIBD is just a BIBD, the condition (\ref{v.eq}) is necessary and sufficient 
for the design to be equitably orderable. This fact follows from Theorem \ref{BIBD.thm}. However, when $c> 1$, it is not known 
if  (\ref{v.eq}) is a sufficient condition for a splitting BIBD to be equitably orderable.
 
 The following result is shown in \cite{PS}.
 
\begin{lemma}
\label{group.lem}
Suppose that a $(v, u \times c,1)$-splitting BIBD is generated by base blocks over an abelian group of order $v$, and suppose every orbit of blocks has size $v$. Then the splitting BIBD can be equitably ordered.
\end{lemma}

\begin{example}
\label{25,3,2.ex} A $(25, 3 \times 2,1)$-splitting BIBD is presented in \cite{GMW}. It has points in $\zed_{25}$ and it is generated from the base block 
\[ \{ \{0,1\}, \{2,4\}, \{12,20\} \}.\]
If we order the base block as \[ ( \{0,1\}, \{2,4\}, \{12,20\} )\] and maintain this ordering as the block is developed, we obtain the blocks
\[
\begin{array}{c}
( \{0,1\}, \{2,4\}, \{12,20\} ) \\
( \{1,2\}, \{3,5\}, \{13,21\} ) \\
\vdots \\
( \{24,0\}, \{1,3\}, \{11,19\} ).
\end{array}
\]
This is an equitable ordering of the splitting BIBD.
\end{example}

It is also shown in \cite{PS} that some infinite families of splitting BIBDs that are constructed recursively can be equitably ordered. Specifically, the cases $u=2$ and $(u,c) = (3,2), (3,3), (3,4)$ and $(4,2)$ are almost completely solved (with a small number of possible exceptions). See \cite{PS} for additional details.

\begin{theorem} 
\label{SplBIBD.thm}
Suppose there is an equitably ordered $(v, u \times c,1)$-splitting BIBD. 
Then there is a $c$-splitting authentication code $E$ for $u$ equiprobable sources, having $v$ messages and 
$b = v(v-1)/(u(u-1)c^2)$ keys, such that
\begin{enumerate}
\item the code provides perfect secrecy,
\item $P_{d_0} = cu/v$ and $P_{d_1} = c(u-1)/(v-1)$, and
\item the optimal key-substitution attack has success probability $1/(cu)$.
\end{enumerate}
\end{theorem}

\begin{proof}
Part 1 follows from Theorem \ref{secrecy.lem} because the splitting BIBD is equitably ordered. Part 2 is shown in \cite[Theorem 5.5]{OKSS}.
Part 3 is proven as follows. Suppose $K$ is the given key. Fix any  message $m \in \mu(K)$. Since the splitting BIBD is equitably ordered, there are 
\[r - \frac{r}{u} = \frac{v-1}{cu}\] keys $K'$ 
such that $m \in e_K(s) \cap e_{K'}(s')$ with $s \neq s'$.
Since $\lambda =1 $, the number of keys 
$K' \neq K$ such that there exists a message $m \in e_K(s) \cap e_{K'}(s')$ with $s \neq s'$ 
is \[ cu \times \frac{v-1}{cu} = v-1.\]
The attacker should replace $K$ by one of these $v-1$ keys. 
Since sources are equiprobable, the key-substitution attack will succeed with probability 
\[ \frac{\frac{v-1}{cu}}{v-1} = \frac{1}{cu}. \]
\end{proof}

Applying Theorem \ref{main.thm}, we have the following.

\begin{theorem}
If there exists an equitably ordered $(v, u \times c,1)$-splitting BIBD, then there exists 
\begin{enumerate}
\item a $c$-splitting authentication code with perfect secrecy for $u$ equiprobable sources that is $(1/cu)$-secure against message-substitution and key-substitution attacks, and 
\item a $(1/cu)$-secure $(2,2)$-threshold scheme for $k$ equiprobable secrets.
\end{enumerate}
\end{theorem}

\begin{proof}
We showed in Theorem \ref{SplBIBD.thm} that the authentication code arising from a $(v, u \times c,1)$-splitting BIBD is 
$c(u-1)/(v-1)$-secure against message-substitution attacks and 
$(1/cu)$-secure against key-substitution attacks. 
In the proof of Theorem \ref{SplBIBD.thm} it is shown that $b \geq v$, so
\[ v-1 \geq u(u-1)c^2,\]
or \[ \frac{c(u-1)}{v-1} \leq \frac{1}{cu}.\]
Hence the authentication code is $(1/cu)$-secure against both attacks and 
the stated result follows directly from  Theorem \ref{main.thm}.
\end{proof}

\section{Dual Authentication Codes}
\label{dual.sec}

Suppose we have an authentication code with sources $\SSS$, messages $\M$, and keyspace $\K$. The encoding matrix is denoted by $E$. Then we can construct another authentication code, which we call the \emph{dual code}, by simply interchanging the roles of messages and keys. Thus, the encoding matrix of the dual code is the matrix $F$ having entries \[F(m,s) = \{K \in \K: m \in e_K(s)\},\] where $s \in \SSS$ and $m \in M$. The keys in the dual code are the same as the messages in the original code. 

It is not hard to see that a key-substitution attack in an authentication code is ``equivalent'' to a message-substitution attack in the dual code. %We record this obvious but useful observation for future use.

\begin{theorem}
\label{dual.thm}
A message-substitution attack in an authentication code is successful if and only if the corresponding key-substitution attack is successful in the dual authentication code. 
\end{theorem}

Note that the probability of a ``key'' in the dual code is the same as the probability of the corresponding message in the original code. Thus, keys in the dual code will be equiprobable if and only if messages in the original code are equiprobable. 
In all the examples we consider, we will assume that condition 3.\ of Theorem \ref{secrecy.lem} holds. This will ensure that a code and its dual both have equiprobable keys and messages. 

Theorem \ref{dual.thm} provides an alternative method to compute success probabilities of key-substitution attacks. 
We illustrate by reconsidering some of the constructions from Section \ref{constructions.sec}, where we computed these success probabilities from first principles. 

If we begin with an authentication code having an encoding matrix that is a $(v,k,\lambda)$-SBIBD, then the rows of the encoding matrix of the dual code, considered as sets, forms the \emph{dual design} of the SBIBD. 
It is a classical result in design theory that the dual design of an SBIBD is again a $(v,k,\lambda)$-SBIBD. 
Thus, Theorem \ref{dual.thm} provides a quick way to see that the optimal success probabilities
of the key-substitution and message-substitution attacks are identical in this particular situation
(as we showed previously in Theorem \ref{SBIBD}).

\begin{example} 
{\rm We return to Example \ref{731.exam}, where we constructed an authentication code from a $(7,3,1)$-SBIBD. We display the  encoding matrices of the code and the dual code:
\[
\begin{array}{c@{\hspace{.5in}}c}
E = 
\begin{array}{|c|c|c|}
\hline
s_1 & s_2 & s_3 \\ \hline 
0 & 1 & 3\\ \hline
1 & 2 & 4\\ \hline
2 & 3 & 5\\ \hline
3 & 4 & 6\\ \hline
4 & 5 & 0\\ \hline
5 & 6 & 1\\ \hline
6 & 0 & 2\\ \hline
\end{array}
& F = \begin{array}{|c|c|c|}
\hline
s_1 & s_2 & s_3 \\ \hline 
K_0 & K_6 & K_4\\ \hline
K_1 & K_0 & K_5\\ \hline
K_2 & K_1 & K_6\\ \hline
K_3 & K_2 & K_0\\ \hline
K_4 & K_3 & K_1\\ \hline
K_5 & K_4 & K_2\\ \hline
K_6 & K_5 & K_3\\ \hline
\end{array}
\end{array}
\]
The rows of $E$ are indexed by $K_0, \dots , K_6$ and the rows of $F$ are indexed by $0, \dots, 6$.
The rows of $F$ comprise the blocks of the dual $(7,3,1)$-SBIBD.
}
\end{example}

Suppose we start with an authentication code $E$ arising from an EDF and then we construct the dual authentication code, $F$.
Let $D_1, \dots , D_k$ be the $c$-subsets in the original EDF. It is not hard to see that the dual authentication code $F$ is generated from the EDF consisting of the $k$ sets $-D_1, \dots , -D_k$. The dual authentication code $F$ satisfies the same properties as $E$ because it is also obtained from an $(n,k,c,\lambda)$-EDF. Thus we see immediately from Theorem \ref{dual.thm} that the success probability of a key-substitution attack in $E$ is $c(k-1)/(n-1)$ (as we showed previously in Theorem \ref{EDF.thm}).

To illustrate, we present a small example.

\begin{example} 
{\rm We have already noted in Example \ref{EDF.exam} that the three sets $\{1,7,11\}$, $\{4,7,9\}$, $\{5,16,17\}$ form a $(19,3,3,3)$-EDF in $\zed_{19}$.
We develop these sets modulo 19 to obtain the following encoding matrices for a $3$-splitting authentication code
and its dual code:
\[E =
\begin{array}{|c|c|c|}
\hline
s_1 & s_2 & s_3 \\ \hline 
\{1,7,11\} & \{4,6,9\} & \{5,16,17\} \\ \hline
\{2,8,12\} & \{5,7,10\} & \{6,17,18\} \\ \hline
\{3,9,13\} & \{6,8,11\} & \{7,18,0\} \\ \hline
\vdots & \vdots & \vdots \\ \hline
\{0,6,10\} & \{3,5,8\} & \{4,15,16\} \\ \hline
\end{array}
\]
\[F = 
\begin{array}{|c|c|c|}
\hline
s_1 & s_2 & s_3 \\ \hline 
\{K_8,K_{12}, K_{18}\} & \{K_{10},K_{13}, K_{15}\} & \{K_2,K_{3}, K_{14}\} \\ \hline
\{K_9,K_{13}, K_{0}\} & \{K_{11},K_{14}, K_{16}\} & \{K_3,K_{4}, K_{15}\} \\ \hline
\{K_{10},K_{14}, K_{1}\} & \{K_{12},K_{15}, K_{17}\} & \{K_4,K_{5}, K_{16}\} \\ \hline
\vdots & \vdots & \vdots \\ \hline
\{K_7,K_{11}, K_{17}\} & \{K_{10},K_{13}, K_{15}\} & \{K_1,K_{2}, K_{13}\} \\ \hline
\end{array}
\]
The rows of $E$ are indexed by $K_0, \dots , K_{18}$ and the rows of $F$ are indexed by $0, \dots, 18$.
We can view $F$ as being generated from the EDF consisting of sets $\{8,12,18\}$, $\{10,13,15\}$, $\{2,3,14\}$.
}
\end{example}

Here is another example, which makes use of a BIBD with $\lambda = 1$ that is not a symmetric BIBD.
 
\begin{example} 
\label{1331.bibd}
{\rm We construct an authentication code from a $(13,3,1)$-BIBD. This design has $r = 6$, and $6 \equiv 0 \bmod 3$, so we can ensure that the corresponding authentication code has perfect secrecy. The $26$ blocks of the design can be generated from the two base blocks $\{0,1,4\}$ and $\{0,2,8\}$ by developing them modulo $13$. The $26$ by $3$  encoding matrix $E$ of the code is as follows:
\[
\begin{array}{c@{\hspace{.5in}}c}
\begin{array}{|r|r|r|}
\hline
s_1 & s_2 & s_3 \\ \hline 
0 & 1 & 4\\ \hline
1 & 2 & 5\\ \hline
2 & 3 & 6\\ \hline
3 & 4 & 7\\ \hline
4 & 5 & 8\\ \hline
5 & 6 & 9\\ \hline
6 & 7 & 10\\ \hline
7 & 8 & 11\\ \hline
8 & 9 & 12\\ \hline
9 & 10 & 0\\ \hline
10 & 11 & 1\\ \hline
11 & 12 & 2\\ \hline
12 & 0 & 3\\ \hline
\end{array}
& 
\begin{array}{|r|r|r|}
\hline
s_1 & s_2 & s_3 \\ \hline 
0 & 1 & 8\\ \hline
1 & 2 & 9\\ \hline
2 & 3 & 10\\ \hline
3 & 4 & 11\\ \hline
4 & 5 & 12\\ \hline
5 & 6 & 0\\ \hline
6 & 7 & 1\\ \hline
7 & 8 & 2\\ \hline
8 & 9 & 3\\ \hline
9 & 10 & 4\\ \hline
10 & 11 & 5\\ \hline
11 & 12 & 6\\ \hline
12 & 0 & 7\\ \hline
\end{array}
\end{array}
\]
The dual code has the following $13$ by $3$ encoding matrix $F$:
\[
\begin{array}{|c|c|c|}
\hline
s_1 & s_2 & s_3 \\ \hline 
\{K_0,K_{13}\} & \{K_{12},K_{25}\} & \{K_9,K_{18}\} \\ \hline
\{K_{1},K_{14}\} & \{K_{0},K_{13}\} & \{K_{10},K_{19}\} \\ \hline
\{K_{2},K_{15}\} & \{K_{1},K_{14}\} & \{K_{11},K_{20}\} \\ \hline
\{K_{3},K_{16}\} & \{K_{2},K_{15}\} & \{K_{12},K_{21}\} \\ \hline
\{K_{4},K_{17}\} & \{K_{3},K_{16}\} & \{K_{0},K_{22}\} \\ \hline
\{K_{5},K_{18}\} & \{K_{4},K_{17}\} & \{K_{1},K_{23}\} \\ \hline
\{K_{6},K_{19}\} & \{K_{5},K_{18}\} & \{K_{2},K_{24}\} \\ \hline
\{K_{7},K_{20}\} & \{K_{6},K_{19}\} & \{K_{3},K_{25}\} \\ \hline
\{K_{8},K_{21}\} & \{K_{7},K_{20}\} & \{K_{4},K_{13}\} \\ \hline
\{K_{9},K_{22}\} & \{K_{8},K_{21}\} & \{K_{5},K_{14}\} \\ \hline
\{K_{10},K_{23}\} & \{K_{9},K_{22}\} & \{K_{6},K_{15}\} \\ \hline
\{K_{11},K_{24}\} & \{K_{10},K_{23}\} & \{K_{7},K_{16}\} \\ \hline
\{K_{12},K_{25}\} & \{K_{11},K_{24}\} & \{K_{8},K_{17}\} \\ \hline
\end{array}
\]
As can be seen, the dual code is $2$-splitting.
The rows of $E$ are indexed by $K_0, \dots , K_{25}$ and the rows of $F$ are indexed by $0, \dots, 12$.
Theorem \ref{BIBD.thm} states that the optimal success probability of a key-substitution attack for $E$ is $1/3$. This is of course the same as the optimal success probability of a message-substitution attack for $F$, by Theorem \ref{dual.thm}.
}
\end{example}

We now explore some additional properties relating authentication codes to their duals.

\begin{theorem}
\label{dualproperties.thm}
Suppose a $c$-splitting authentication code for $u$ sources has $b$ equiprobable keys, equiprobable message encoding,
$v$ messages, perfect secrecy, and  $P_{d_0} = cu/v$. %., and each message occurs in exactly $r$ sets $\mu(K)$.
Then the dual authentication code is a $(bc/v)$-splitting authentication code for $u$ sources that has $v$ equiprobable keys,  equiprobable message encoding, $b$ messages, perfect secrecy, and  $P_{d_0} = cu/v$. %Further, each message occurs in exactly $cu$ sets $\mu(m)$.
\end{theorem}

\begin{proof} 
The proof of Theorem \ref{secrecy.lem} establishes that equation (\ref{kappa.eq}) holds, i.e., every message $m$ occurs $bc/v$ times in each column $s$ of the encoding matrix $E$. 
This immediately implies that the dual code is $(bc/v)$-splitting. 
Therefore the dual code is a $(bc/v)$-splitting authentication code for $u$ sources having $v$ equiprobable keys and  equiprobable message encoding. Each ``message'' in the dual code occurs $c$ times in each column $s$ of $F$ (where $F$ is the the encoding matrix of the dual code). 
Therefore, from Theorem \ref{secrecy.lem}, the  dual code has perfect secrecy and 
\[ P_{d_0} = \frac{\frac{bc}{v} \times u}{b} = \frac{cu}{v}.\]
\end{proof}

We note that the hypotheses of Theorem \ref{dualproperties.thm} are satisfied 
whenever we construct an authentication code from an equitably ordered BIBD or splitting BIBD.

The authentication code presented in Example \ref{1331.bibd} satisfies the hypotheses of  Theorem \ref{dualproperties.thm} with  $v = 13$, $b = 26$, $u = 3$, $c=1$.  Thus, the dual code is $2$-splitting with perfect secrecy, each ``message'' $K_i$ occurs once in each column of $F$. The code and dual code both have $P_{d_0} = 3/13$.

\section{Summary and Discussion}
\label{summary.sec}

Our goal in this paper has been to develop some theory to better understand various connections between authentication codes and threshold schemes, as well as how certain combinatorial designs can be used to construct these cryptographic objects. To this end, we have proven a simple direct equivalence of certain authentication codes and $(2,2)$-threshold schemes. Further, we have introduced the notion of a key-substitution attack and observed that it is identical to a message-substitution attack in a ``dual authentication code.''

We have already mentioned that robust $(k,n)$-threshold schemes are usually constructed by ``combining'' an algebraic object such as a difference set, EDF, or AMD code with a Shamir threshold scheme. These objects all live in a finite group and, consequently, the construction of the resulting threshold schemes is algebraic. The main equivalence result we have proven (Theorem \ref{main.thm}) is a purely combinatorial result. It would be of interest to extend our equivalence theorem in some way to handle robust $(k,n)$-threshold schemes in a strictly combinatorial setting.
There is a purely combinatorial analogue of Shamir threshold schemes---namely, orthogonal arrays---so this is perhaps possible.

\end{document}